\documentclass[graybox]{svmult}
\usepackage{mathptmx}       % selects Times Roman as basic font
\usepackage{helvet}         % selects Helvetica as sans-serif font
\usepackage{courier}        % selects Courier as typewriter font
\usepackage{makeidx}         % allows index generation
\usepackage{graphicx}        % standard LaTeX graphics tool
\usepackage{multicol}        % used for the two-column index
\usepackage[bottom]{footmisc}% places footnotes at page bottom
\usepackage{amsmath,amssymb,amsfonts}        % AMS math packages
\usepackage{empheq}
\makeindex             % used for the subject index
\DeclareMathAlphabet{\mathcal}{OMS}{cmsy}{m}{n}
\newcommand{\R}{\mathbb{R}}
\renewcommand{\E}{\mathbb{E}}
\DeclareMathOperator*{\argmin}{arg\,min}
\DeclareMathOperator*{\argmax}{arg\,max}

\begin{document}
\title*{Local Volatility Calibration by Optimal Transport}
\author{Ivan Guo$^{*\dagger}$ \and Gr\'egoire Loeper$^{*\dagger}$ \and Shiyi Wang$^{*}$}
\institute{
$^*$ School of Mathematical Sciences, Monash University, Australia\\
$^\dagger$ Centre for Quantitative Finance and Investment Strategies, Monash University, Australia\\
\textbf{Acknowledgements\quad} The Centre for Quantitative Finance and Investment Strategies has been supported by BNP Paribas.
}
\maketitle

\abstract{The calibration of volatility models from observable option prices is a fundamental problem in quantitative finance. The most common approach among industry practitioners is based on the celebrated Dupire's formula \cite{Dupire94pricing}, which requires the knowledge of vanilla option prices for a continuum of strikes and maturities that can only be obtained via some form of price interpolation. In this paper, we propose a new local volatility calibration technique using the theory of optimal transport. We formulate a time continuous martingale optimal transport problem, which seeks a martingale diffusion process that matches the known densities of an asset price at two different dates, while minimizing a chosen cost function. Inspired by the seminal work of Benamou and Brenier \cite{MR1738163}, we formulate the problem as a convex optimization problem, derive its dual formulation, and solve it numerically via an augmented Lagrangian method and the alternative direction method of multipliers (ADMM) algorithm. The solution effectively reconstructs the dynamic of the asset price between the two dates by recovering the optimal local volatility function, without requiring any time interpolation of the option prices.
}

%%%%%%%%%%%%%%%%%%%%%%%%%%%%%%%%%%%%%%%%%%%%%%%%%%%%%%%%%%%%%%%%%%%%
\section{Introduction}
\label{sec:introduction}

A fundamental assumption of the classical Black-Scholes option pricing framework is that the underlying risky asset has a constant volatility. However, this assumption can be easily dispelled by the option prices observed in the market, where the implied volatility surfaces are known to exhibit ``skews'' or ``smiles''. Over the years, many sophisticated volatility models have been introduced to explain this phenomenon. One popular class of model is the local volatility models. In a local volatility model, the volatility function $\sigma(t,S_t)$ is a function of time $t$ as well as the asset price $S_t$. The calibration of the local volatility function involves determining $\sigma$ from available option prices. 

One of the most prominent approaches for calibrating local volatility is introduced by the path-breaking work of Dupire \cite{Dupire94pricing}, which provides a method to recover the local volatility function $\sigma(t,s)$ if the prices of European call options $C(T,K)$ are known for a continuum of maturities $T$ and strikes $K$. In particular, the famous Dupire's formula is given by
\begin{equation}
  \sigma^2(T,K) = \frac{\frac{\partial C(T,K)}{\partial T}+\mu_t K\frac{\partial C(T,K)}{\partial K}}{\frac{K^2}{2}\frac{\partial^2 C(T,K)}{\partial K^2}},
\end{equation}
where $\mu_t$ is a deterministic function. However, in practice, option prices are only available at discrete strikes and maturities, hence interpolation is required in both variables to utilize this formula, leading to many inaccuracies. Furthermore, the numerical evaluation of the second derivative in the denominator can potentially cause instabilities in the volatility surface as well as singularities. Despite these drawbacks, Dupire's formula and its variants are still used prevalently in the financial industry today. 

In this paper, we introduce a new technique for the calibration of local volatility functions that adopts a variational approach inspired by optimal transport.
The optimal transport problem was first proposed by Monge \cite{monge1781memoire} in 1781 in the context of civil engineering. The basic problem is to transfer material from one site to another while minimizing transportation cost. In the 1940's, Kantorovich \cite{MR2117877} provided a modern treatment of the problem based on linear programming techniques, leading to the so-called Monge-Kantorovich problem. Since then, the theory of optimal transport has attracted considerable attention with applications in many areas such as fluid dynamics, meteorology, econometrics and cosmology (see, e.g., \cite{MR1698853}, \cite{MR2209129} and \cite{MR2459454}). Recently, there have been a few studies extending optimal transport to stochastic settings with applications in financial mathematics. For instance, Tan and Touzi \cite{MR3127880} studied an extension of the Monge-Kantorovich problem for semimartingales, while Dolinsky and Soner \cite{MR3256817} applied martingale optimal transport to the problem of robust hedging.

In our approach, we begin by recovering the probability density of the underlying asset at times $t_0$ and $t_1$ from the prices of European options expiring at $t_0$ and $t_1$. Then, instead of interpolating between different maturities, we seek a martingale diffusion process which transports the density from $t_0$ to $t_1$, while minimizing a particular cost function. This is similar to the classical optimal transport problem, with the additional constraint that the diffusion process must be a martingale driven by a local volatility function. In the case where the cost function is convex, the problem can be reformulated as a convex optimization problem under linear constraints, and this problem has been recently studied by Huesmann et al.\! in \cite{2017arXiv170701493H} and \cite{2017arXiv170804869B}. Theoretically, the stochastic control problem can be reformulated as an optimization problem which involves solving a non-linear PDE at each step, and the PDE is closely connected with the ones studied in Bouchard et al.\! \cite{MR3715374, MR3519167} and Loeper \cite{2013arXiv1301.6252L} in the context of option pricing with market impact. For this paper, we approach the problem via the augmented Lagrangian method and the alternative direction method of multipliers (ADMM) algorithm, which was also used in Benamou and Brenier \cite{MR1738163} for classical optimal transport problems. 

The paper is organized as follows. In Section 2, we introduce the classical optimal transport problem as formulated by Benamou and Brenier \cite{MR1738163}. In Section 3, we introduce the martingale optimal transport problem and its augmented Lagrangian. The numerical method is detailed in Section 4 and numerical results are given in Section 5.

%%%%%%%%%%%%%%%%%%%%%%%%%%%%%%%%%%%%%%%%%%%%%%%%%%%%%%%%%%%%%%%%%%%%

\section{Optimal Transport}
\label{sec:optimal_transport}
In this section, we briefly outline the optimal transport problem as formulated by Benamou and Brenier \cite{MR1738163}.
Given density functions $\rho_0, \rho_1: \R^d \to [0,\infty)$ with equal total mass $\int_{\R^d} \rho_0(x) \D x = \int_{\R^d} \rho_1(x) \D x$. We say that a map $s:\R^d\to\R^d$ is an admissible transport plan if it satisfies
\begin{equation}\label{eq:balance}
  \int_{x\in A}\rho_1(x)\D x = \int_{s(x)\in A}\rho_0(x)\D x,
\end{equation}
for all bounded subset $A\subset\R^d$. Let $\mathcal{T}$ denote the collection of all admissible maps. Given a cost function $c(x,y)$, which represents the transportation cost of moving one unit of mass from $x$ to $y$, the optimal transport problem is to find an optimal map $s^*\in\mathcal{T}$ that minimizes the total cost
\begin{equation}
  \inf_{s\in\mathcal{T}}\int_{\R^d}c(x, s(x)) \rho_0(x) \D x.
\end{equation}
In particular, when $c(x,y) = |y-x|^2$ where $|\cdot|$ denotes the Euclidean norm, this problem is known as the $L^2$ Monge-Kantorovich problem (MKP). 

The $L^2$ MKP is reformulated in \cite{MR1738163} in a fluid mechanic framework. In the time interval $t\in [0,1]$, consider all possible smooth, time-dependent, densities $\rho(t,x)\geq 0$ and velocity fields $v(t,x)\in\R^d$, that satisfy the continuity equation
\begin{equation}\label{eq:continuity_equation}
  \partial_t \rho(t,x) + \nabla\cdot(\rho(t,x) v(t,x))=0, \qquad \forall t\in[0,1],\;\forall x\in\R^d,
\end{equation}
and the initial and final conditions
\begin{equation}\label{eq:initial_final_condition}
  \rho(0,x)=\rho_0 ,\quad \rho(1,x)=\rho_1.
\end{equation}
In \cite{MR1738163}, it is proven that the $L^2$ MKP is equivalent to finding an optimal pair $(\rho^*,v^*)$ that minimizes
\begin{equation}
  \inf_{\rho,v}\int_{\R^d}\int^1_0 \rho(t,x)|v(t,x)|^2 \D t \D x,
\end{equation}
subject to the constraints (\ref{eq:continuity_equation}) and (\ref{eq:initial_final_condition}). This problem is then solved numerically in \cite{MR1738163} via an augmented Lagrangian approach. The specific numerical algorithm used is
 %augmented Langragian method This algorithm was originally called ALG2 in \cite{MR1738163}, but it is now 
known as the alternative direction method of multipliers (ADMM), which has applications in statistical learning and distributed optimization.

%%%%%%%%%%%%%%%%%%%%%%%%%%%%%%%%%%%%%%%%%%%%%%%%%%%%%%%%%%%%%%%%%%%%
\section{Definition of the martingale problem}
\label{sec:setting}
Let $(\Omega, \mathbb{F}, \mathbb{Q})$ be a probability space, where $\mathbb{Q}$ is the risk-neutral measure. Suppose the dynamic of an asset price $X_t$ on $t\in[0,1]$ is given by the local volatility model
\begin{equation}\label{eq:diffusion_process_2}
  \D X_t = \sigma(t,X_t)\D W_t, \quad t\in[0,1],
\end{equation}
where $\sigma(t,x)$ is a local volatility function and $W_t$ is a one-dimensional Brownian motion. For the sake of simplicity, suppose the interest and dividend rates are zero.
Denote by $\rho(t,x)$ the density function of $X_t$ and $\gamma(t,x)=\sigma(t,x)^2/2$ the diffusion coefficient. It is well known that $\rho(t,x)$ follows the Fokker-Planck equation
\begin{equation}\label{eq:fokker_planck}
  \partial_t\rho(t,x) -\partial_{xx}(\rho(t,x)\gamma(t,x)) = 0.
\end{equation}
Suppose that the initial and the final densities are given by $\rho_0(x)$ and $\rho_1(x)$, which are recovered from European option prices via the Breeden-Litzenberger \cite{breeden1978prices} formula,
\[
\rho_T(K) = \frac{\partial^2 C(T,K)}{\partial K^2}.
\]
%which can be approximated by finite difference method.

Let $F:\R\to\R\cup\{+\infty\}$ be a convex cost function. We are interested in minimizing the quantity
\[
\E\left(\int^1_0 F\left(\gamma(t,X_t)\right) \D t \right)=\int_D\int^1_0 \rho(t,x)F\left(\gamma(t,X_t)\right) \D t \D x,
\]
where $F(x)=+\infty$ if $x<0$, and $D\subseteq \R$ is the support of $\{X_t, t\in[0,1]\}$. Unlike the classical optimal transport problem, the existence of a solution here requires an additional condition: there exists a martingale transport plan if and only if $\rho_0$ and $\rho_1$ satisfy:
\[
  \int_\R\varphi(x)\rho_0(x)\D x \leq \int_\R\varphi(x)\rho_1(x)\D x,
\]
for all convex function $\varphi(x):\R\to\R$. This is known as Strassen's Theorem \cite{MR0177430}. This condition is naturally satisfied by financial models in which the asset price follows a martingale diffusion process.

%This condition is usually satisfied in the financial world. However, there doesn't exist a martingale transport plan if the dynamics of the asset price is not a diffusion process.

\begin{remark}
The formulation here is actually quite general and it can be easily adapted to a large family of models. For example, the case of a geometric Brownian motion with local volatility can be recovered by substituting
$\tilde\sigma(t,X_t)X_t=\sigma(t,X_t)$ everywhere, including in the Fokker-Planck equation. The cost function $F$ would then also be dependent on $x$. The later arguments involving convex conjugates still hold since $F$ remains a convex function of $\tilde\sigma$.
\end{remark}

Since $\rho F(\gamma)$ is not convex in $(\rho,\gamma)$ (which is crucial for our method), the substitution 
$m(t,x):=\rho(t,x)\gamma(t,x)$ is applied.
So we obtain the following \emph{martingale optimal transport problem}:
\begin{equation}\label{eq:objective}
  \inf_{\rho,m} \int_D\int^1_0 \rho(t,x)F\left(\frac{m(t,x)}{\rho(t,x)}\right) \D t \D x,
\end{equation}
subject to the constraints:
\begin{gather}
\rho(0,x) = \rho_0(x) ,\quad \rho(1,x) = \rho_1(x),\label{eq:initial_final}\\
 \partial_t\rho(t,x) -\partial_{xx}m(t,x) = 0.\label{eq:fokker_planck2}
\end{gather}
Using the convexity of $F$, the term $\rho F(m/\rho)$ can be easily verified to be convex in $(\rho,m)$. Also note that we have the natural restrictions of $\rho>0$ and $m \geq 0$. Note that $m\geq 0$ is enforced by penalizing the cost function $F$, and $\rho>0$ will be encoded in the convex conjugate formulation. (see Proposition \ref{propap02})

Next, introduce a time-space dependent Lagrange multiplier $\phi(t,x)$ for the constraints \eqref{eq:initial_final} and \eqref{eq:fokker_planck2} . Hence the associated Lagrangian is 
\begin{equation}\label{eq:Lagrangian_1}
  L(\phi, \rho, m) = \int_\R\int_0^1 \rho(t,x)F\left(\frac{m(t,x)}{\rho(t,x)}\right) + \phi(t,x)\big(\partial_t\rho(x)-\partial_{xx}(m(t,x))\big) \D t \D x.
\end{equation}
Integrating \eqref{eq:Lagrangian_1} by parts and letting $m=\rho\gamma$ vanish on the boundaries of $D$, the martingale optimal transport problem can be reformulated as the following saddle point problem:
\begin{align}\label{eq:Lagrangian_2}
    \inf_{\rho,m}\sup_\phi L(\phi, \rho, m) &=   \inf_{\rho,m}\sup_\phi\int_D\int_0^1 \left(\rho F\left(\frac{m}{\rho}\right) - \rho\partial_t\phi - m\partial_{xx}\phi\right)  \D t \D x \nonumber\\
	&\qquad- \int_D\left(\phi(0,x)\rho_0-\phi(1,x)\rho_1\right)\D x.
\end{align}
%Thus the martingale optimal transport problem can be reformulated as the following saddle point problem:
%\begin{equation}\label{eq:Lagrangian_3}
  %\inf_{\rho,m}\sup_\phi L(\phi,\rho,m).
%\end{equation}
As shown by Theorem 3.6 in \cite{MR3127880}, (\ref{eq:Lagrangian_2}) has an equivalent dual formulation which leads to the following representation:
\begin{align}
	\sup_\phi\inf_{\rho,m} L(\phi,\rho,m)&=\sup_\phi\inf_{\rho} \int_D\int_0^1 -\rho\left(\partial_t \phi+F^*(\partial_{xx}\phi)\right)  \D t \D x \nonumber\\
	&\qquad- \int_D\left(\phi(0,x)\rho_0-\phi(1,x)\rho_1\right)\D x.
\end{align}
In particular, the optimal $\phi$ must satisfy the condition
%, we know that (\ref{eq:Lagrangian_2}) is equivalent to another Lagrangian which has $\rho$ as the multiplier and the following a Hamilton--Jacobi--Bellman (HJB) equation as the constraint:
\begin{equation}\label{eq:HJB}
  \partial_t \phi+F^*(\partial_{xx}\phi)=0,
\end{equation}
where $F^*$ is the convex conjugate of $F$ (see \eqref{eq:lagendre_transform} and Proposition \ref{propap02}). We will later use (\ref{eq:HJB}) to check the optimality of our algorithm.

% The optimality conditions can be obtained by setting the functional derivatives of $L(\phi, \rho, m)$ with respect to $\phi$, $\rho$ and $m$ to zero, resulting in the following equations:
% \begin{empheq}[left=\empheqlbrace]{align}
%   \partial_t\rho - \partial_{xx}m &= 0, \label{eq:opt_con_1}\\
%   F\left(\frac{m}{\rho}\right)-\frac{m}{\rho}F'\left(\frac{m}{\rho}\right)&=\partial_t \phi, \label{eq:opt_con_2}\\
%   F'\left(\frac{m}{\rho}\right)&=\partial_{xx}\phi. \label{eq:opt_con_3}
% \end{empheq}
% Note that \eqref{eq:opt_con_2} and \eqref{eq:opt_con_3} imply the following condition,
% \[
% 	\partial_t \phi+F^*(\partial_{xx}\phi)=0,
% 	\]
% where $F^*$ is the convex conjugate of $F$ (see \eqref{eq:lagendre_transform} and Proposition \ref{propap02}).

\subsubsection*{Augmented Lagrangian Approach}
Similar to \cite{MR1738163}, we solve the martingale optimal transport problem using the augmented Lagrangian approach. Let us begin by briefly recalling the well-known definition and properties of the convex conjugate. For more details, the readers are referred to Section 12 of Rockafellar \cite{MR0274683}.

Fix $D\subseteq\R^d$, let $f:\R^d\to\R\cup\{+\infty\}$ be a proper convex and lower semi-continuous function. Then the \emph{convex conjugate} of $f$ is the function $f^*:\R^d\to\R\cup\{+\infty\}$ defined by 
\begin{equation}\label{eq:lagendre_transform}
f^*(y):=\sup_{x\in \R^d}(x\cdot y-f(x)).
\end{equation}
The convex conjugate is also often known as the \emph{Legendre-Fenchel transform}.

\begin{proposition}\label{propap02}
We have the following properties:\\
(i) $f^*$ is a proper convex and lower semi-continuous function with $f^{**}\equiv f$;\\
(ii) if $f$ is differentiable, then $f(x)+f^*(f'(x))=xf'(x)$.
\end{proposition}

Returning to the problem at hand, recall that $G(x,y) := xF(y/x), x>0$ is convex in $(x,y)$. By adopting the convention of $G(x,y) =\infty$ whenever $x\leq 0$, it can be expressed in terms of the convex conjugate, as shown in the following proposition.
\begin{proposition}\label{propap03} Denote by $F^*$ the convex conjugate of $F$.

(i) Let $G(x,y) = xF(y/x)$, the convex conjugate of $G$ is given by:
\begin{equation}\label{eq:lagendre_transform_G}
  G^*(a,b) = \begin{cases}
    0,  &\text{if } a+F^*(b)\leq 0,\\
    \infty, &\text{otherwise}.
  \end{cases}
\end{equation}

(ii) For $x>0$, We have the following equality,
\begin{equation}
  xF\left(\frac{y}{x}\right)= \sup_{(a,b)\in\R^2} \{ax+b y : a+F^*(b)\leq 0\}.
\end{equation}

\end{proposition}
\begin{proof} (i)
  By definition, the convex conjugate of $G$ is given by
  \begin{align}
    G^*(a,b) &= \sup_{(x,y)\in \R^2} \left\{ ax+by-xF\left(\frac{y}{x}\right) : x>0 \right\} \\
        &= \sup_{(x,y)\in \R^2} \left\{ ax+x\left( b\frac{y}{x}-F\left(\frac{y}{x}\right) \right) : x>0 \right\} \\
		&=\sup_{x>0} \left\{ x(a+F^*(b)) \right\},
  \end{align}
  If $a+F^*(b)\leq 0$, the supremum is achieved by limit $x\to0$, otherwise, $G^*$ becomes unbounded as $x$ increases. This establishes part (i).

(ii) The required equality follows immediately from part (i) and the fact that
\[
xF\left(\frac{y}{x}\right)= \sup_{(a,b)\in\R^2} \{ax+b y - G^*(a,b) : a+F^*(b)\leq 0\}. \tag*{\qed}
\]
%\hfill\hfill\qed
\end{proof}

Now we are in a position to present the augmented Lagrangian. First, let us introduce the following notations:
\begin{gather}
  K = \Bigl\{ ( a,b ) : \R\times\R\to\R\times\R ~\Big|~ a+F^*(b)\leq 0 \Bigr\}, \\
  \mu = (\rho, m) = (\rho, \rho\gamma) ,\quad q=(a,b) ,\quad \langle\mu,q\rangle = \int_D\int_0^1 \mu\cdot q ,\\
  H(q) = G^*(a,b) = \begin{cases}
    0,  &\text{if } q\in K,\\
    \infty, &\text{otherwise},
  \end{cases}\\
  J(\phi) = \int_D[\phi(0,x)\rho_0 - \phi(1,x)\rho_1],\\
  \nabla_{t,xx} = (\partial_t,\partial_{xx}) .
\end{gather}
By using the above notations, we can express the equality from Proposition \ref{propap03} (ii) in the following way,
\begin{equation}
  \rho F\left(\frac{m}{\rho}\right) = \sup_{\{a,b\}\in K} \{a\rho+b m\} = \sup_{q\in K} \{\mu\cdot q\}.
\end{equation}
Since the restriction $q\in K$ is checked point-wise for every $(t,x)$, we can exchange the supremum with the integrals in the following equality
\begin{equation}
  \int_D\int_0^1 \sup_{q\in K} \{\mu\cdot q\} = \sup_q\Bigl\{ -H(q)+\int_D\int_0^1 \mu\cdot q \Bigr\} = \sup_q\Bigl\{ -H(q)+ \langle\mu,q\rangle \Bigr\}.
\end{equation}
Therefore, the saddle point problem specified by \eqref{eq:Lagrangian_2} can be rewritten as
\begin{equation}\label{eq:new_saddle}
  \sup_\mu\inf_{\phi,q}\Bigl\{ H(q)+J(\phi)+\langle\mu,\nabla_{t,xx}\phi - q\rangle \Bigr\}.
\end{equation}
Note that in the new saddle point problem (\ref{eq:new_saddle}), $\mu$ is the Lagrange multiplier of the new constraint $\nabla_{t,xx}\phi = q$. In order to turn this into a convex problem, we define the augmented Lagrangian as follows:
\begin{equation}\label{eq:new_saddle2}
  L_r(\phi,q,\mu)=H(q)+J(\phi)+\langle\mu,\nabla_{t,xx}\phi - q\rangle +\frac{r}{2}\langle\nabla_{t,xx}\phi - q,\nabla_{t,xx}\phi - q\rangle,
\end{equation}
where $r>0$ is a penalization parameter. The saddle point problem then becomes
\begin{equation}\label{eq:new_saddle3}
  \sup_\mu\inf_{\phi,q}L_r(\phi,q,\mu),
\end{equation}\enlargethispage{1\baselineskip}%
which has the same solution as (\ref{eq:Lagrangian_2}).

%%%%%%%%%%%%%%%%%%%%%%%%%%%%%%%%%%%%%%%%%%%%%%%%%%%%%%%%%%%%%%%%%%%%%%%%%%%%
\section{Numerical Method}
\label{sec:numerical_method}

In this section, we describe in detail the alternative direction method of multipliers (ADMM) algorithm for solving the saddle point problem given by \eqref{eq:new_saddle2} and \eqref{eq:new_saddle3}.
In each iteration, using $(\phi^{n-1}, q^{n-1}, \mu^{n-1})$ as a starting point, the ADMM algorithm performs the following three steps:
\begin{alignat}{3}
  \text{Step A: }&\quad \phi^n &&= \argmin_\phi L_r( \phi,q^{n-1},\mu^{n-1} ), \\
  \text{Step B: }&\quad q^n &&= \argmin_q L_r( \phi^n,q,\mu^{n-1} ), \\
  \text{Step C: }&\quad \mu^n &&= \argmax_\mu L_r( \phi^n,q^n,\mu ).
\end{alignat}

\noindent
{\bf Step A:} $\phi^n = \argmin_\phi L_r( \phi,q^{n-1},\mu^{n-1} )$\\

\noindent To find the function $\phi^n$ that minimizes $L_r(\phi,q^{n-1},\mu^{n-1})$, we set the functional derivative of $L_r$ with respect to $\phi$ to zero:
\begin{equation}
  J(\phi) + \langle\mu^{n-1}, \nabla_{t,xx} \phi\rangle + r\langle\nabla_{t,xx}\phi^n-q^{n-1}, \nabla_{t,xx}\phi\rangle = 0 .
\end{equation}
By integrating by parts, we arrive at the following variational equation
\begin{equation}\label{eq:PDE}
  -r(\partial_{tt}\phi^n - \partial_{xxxx}\phi^n) = \partial_t(\rho^{n-1}-ra^{n-1}) - \partial_{xx}(m^{n-1}-rb^{n-1}), 
\end{equation}
with Neumann boundary conditions in time $\forall x\in D$:
\begin{align}
  r\partial_t\phi^n(0,x) &= \rho_0 - \rho^{n-1}(0,x) + ra^{n-1}(0,x),\\
  r\partial_t\phi^n(1,x) &= \rho_1 - \rho^{n-1}(1,x) + ra^{n-1}(1,x).
\end{align}
For the boundary conditions in space, let $D=[\underline{D},\overline{D}]$. We give the following boundary condition to the diffusion coefficient:
\[
  \gamma(t,\underline{D}) = \gamma(t,\overline{D}) = \overline{\gamma} := \argmin_{\gamma\in\R} F(\gamma).
\]
From (\ref{eq:Lagrangian_2}) and (\ref{eq:HJB}), we know $\partial_{xx}\phi$ is the dual variable of $\gamma$. Since $\overline{\gamma}$ minimizes $F$, the corresponding $\partial_{xx}\phi$ must be zero. Therefore, we have the following boundary conditions:
\begin{equation}
  \partial_{xx}\phi(t,\underline{D}) = \partial_{xx}\phi(t,\overline{D}) = 0, \quad \forall t\in[0,1].
\end{equation}
% By setting the functional derivative of (\ref{eq:Lagrangian_2}) with respect to $m$ to zero, we have the following equality,
% \begin{equation}
%   \partial_{xx}\phi(t,\underline{D}) = \partial_{xx}\phi(t,\overline{D}) = F'(0), \quad \forall t\in[0,1].
% \end{equation}
In \cite{MR1738163}, periodic boundary conditions were used in the spatial dimension and a perturbed equation was used to yield a unique solution. Since periodic boundary conditions are inappropriate for martingale diffusion and we are dealing with a bi-Laplacian term in space, we impose the following additional boundary conditions in order to enforce a unique solution:
\begin{equation}
  \phi(t,\underline{D})=\phi(t,\overline{D})=0 ,\quad \forall t\in[0,1].
\end{equation}
Now, the 4th order linear PDE (\ref{eq:PDE}) can be numerically solved by the finite difference method or the finite element method. \\

\noindent
{\bf Step B:} $q^n = \argmin_q L_r( \phi^n,q,\mu^{n-1} )$\\

\noindent Since $H(q)$ is not differentiable, we cannot differentiate $L_r$ with respect to $q$. Nevertheless, we can simply obtain $q^n$ by solving the minimization problem
\begin{equation}
  \inf_q L_r(\phi^n,q,\mu^{n-1}).
\end{equation}
This is equivalent to solving
\begin{equation}
  \inf_{q\in K} \left\langle \nabla_{t,xx}\phi^n + \frac{\mu^{n-1}}{r}-q, \nabla_{t,xx}\phi^n + \frac{\mu^{n-1}}{r}-q \right\rangle.
\end{equation}
Now, let us define
\begin{equation}
  p^n(t,x) = \{\alpha^n(t,x),\beta^n(t,x)\} = \nabla_{t,xx}\phi^n(t,x) + \frac{\mu^{n-1}(t,x)}{r},
\end{equation}
then we can find $q^n(t,x)=\{a^n(t,x),b^n(t,x)\}$ by solving
\begin{equation}\label{eq:stepb_quad}
  \inf_{\{a,b\}\in \R\times\R} \Bigl\{ (a(t,x)-\alpha^n(t,x))^2 + (b(t,x) - \beta^n(t,x))^2 : a+F^*(b) \leq 0 \Bigr\}
\end{equation}
point-wise in space and time. This is a simple one-dimensional projection problem. If $\{\alpha^n,\beta^n\}$ satisfies the constraint $ \alpha^n+F^*(\beta^n) \leq 0$, then it is also the minimum. Otherwise, the minimum must occur on the boundary $a+F^*(b) = 0 $. In this case we substitute the condition into (\ref{eq:stepb_quad}) to obtain
\begin{equation}\label{eq:stepb_poly}
  \inf_{b\in\R} \Bigl\{ (F^*(b(t,x)) + \alpha(t,x))^2+(b(t,x)-\beta(t,x))^2 \Bigr\},
\end{equation}
which must be solved point-wise.
The minimum of \eqref{eq:stepb_poly} can be found using standard root finding methods such as Newton's method. In some simple cases it is even possible to compute the solution analytically.\\

\noindent
{\bf Step C:} $\mu^n = \argmax_\mu L_r( \phi^n,q^n,\mu )$\\

\noindent Begin by computing the gradient by differentiating the augmented Lagrangian $L_r$ respect to $\mu$. Then, simply update $\mu$ by moving it point-wise along the gradient as follows,
\begin{equation}
  \mu^{n}(t,x) = \mu^{n-1}(t,x) + r(\nabla_{t,xx}\phi^n(t,x) - q^n(t,x)).
\end{equation}\\

\newpage
\noindent
{\bf Stopping criteria:}\\

\noindent Recall the HJB equation (\ref{eq:HJB}):
\begin{equation}\label{eq:stop01}
  \partial_t \phi+F^*(\partial_{xx}\phi)=0 .
\end{equation}
%Note that \eqref{eq:stop01} holds even if the optimum occurs at $m=0$.
We use \eqref{eq:stop01} to check for optimality. Define the residual:
\begin{equation}
  res^n=\max_{t\in[0,1],x\in D} \rho \left| \partial_t \phi+F^*(\partial_{xx}\phi) \right|.
\end{equation}
This quantity converges to 0 when it approaches the optimal solution of the problem. The residual is weighted by the density $\rho$ to alleviate any potential issues caused by small values of $\rho$.

%%%%%%%%%%%%%%%%%%%%%%%%%%%%%%%%%%%%%%%%%%%%%%%%%%%%%%%%%%%%%%%%
\section{Numerical Results}
\label{sec:numerical_results}

\begin{figure}[t]
\includegraphics[width=\textwidth]{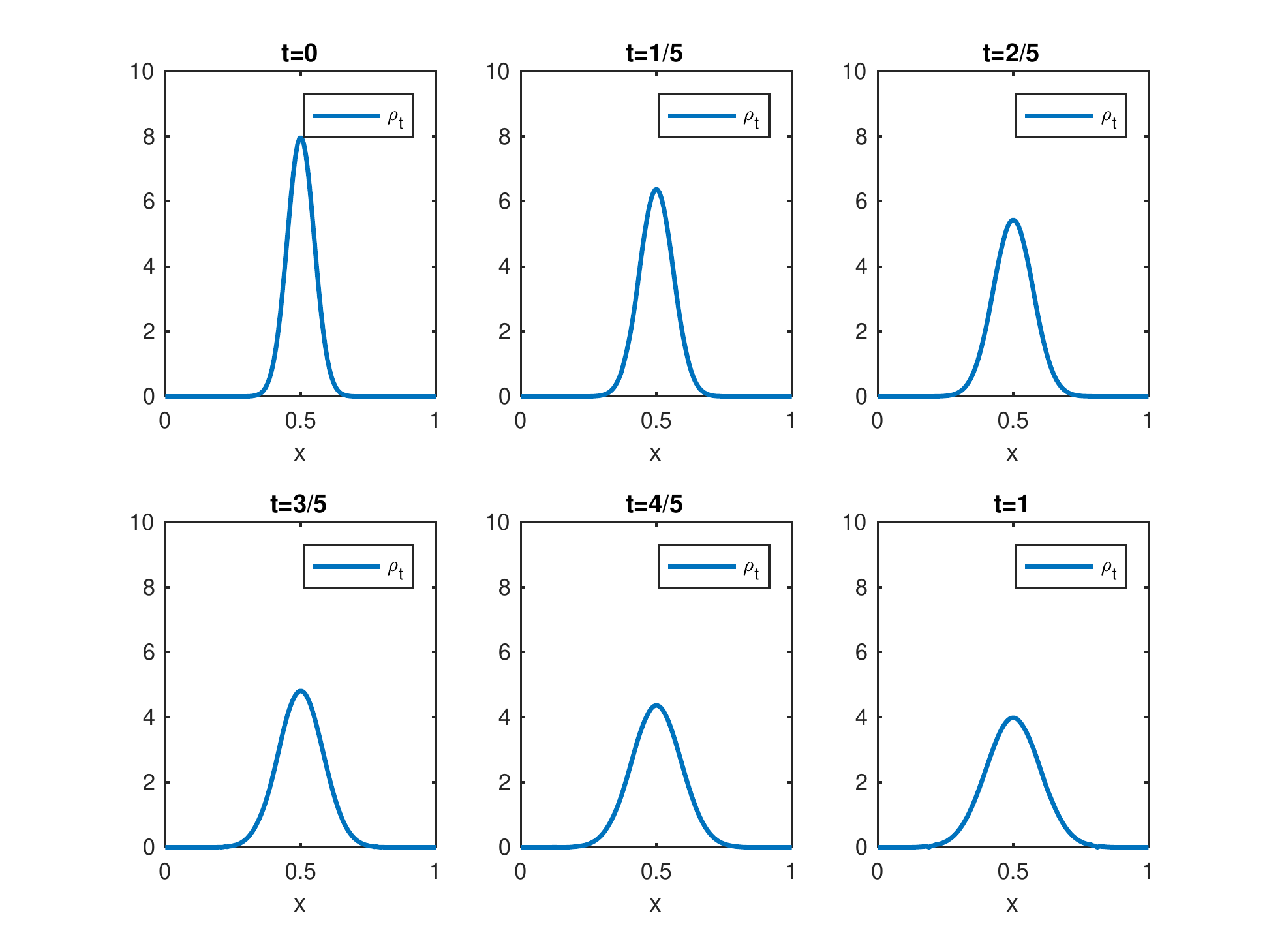}
\caption{The density function $\rho(t,x)$.}
\label{fig1}
\end{figure}
\begin{figure}[h]
\includegraphics[width=\textwidth]{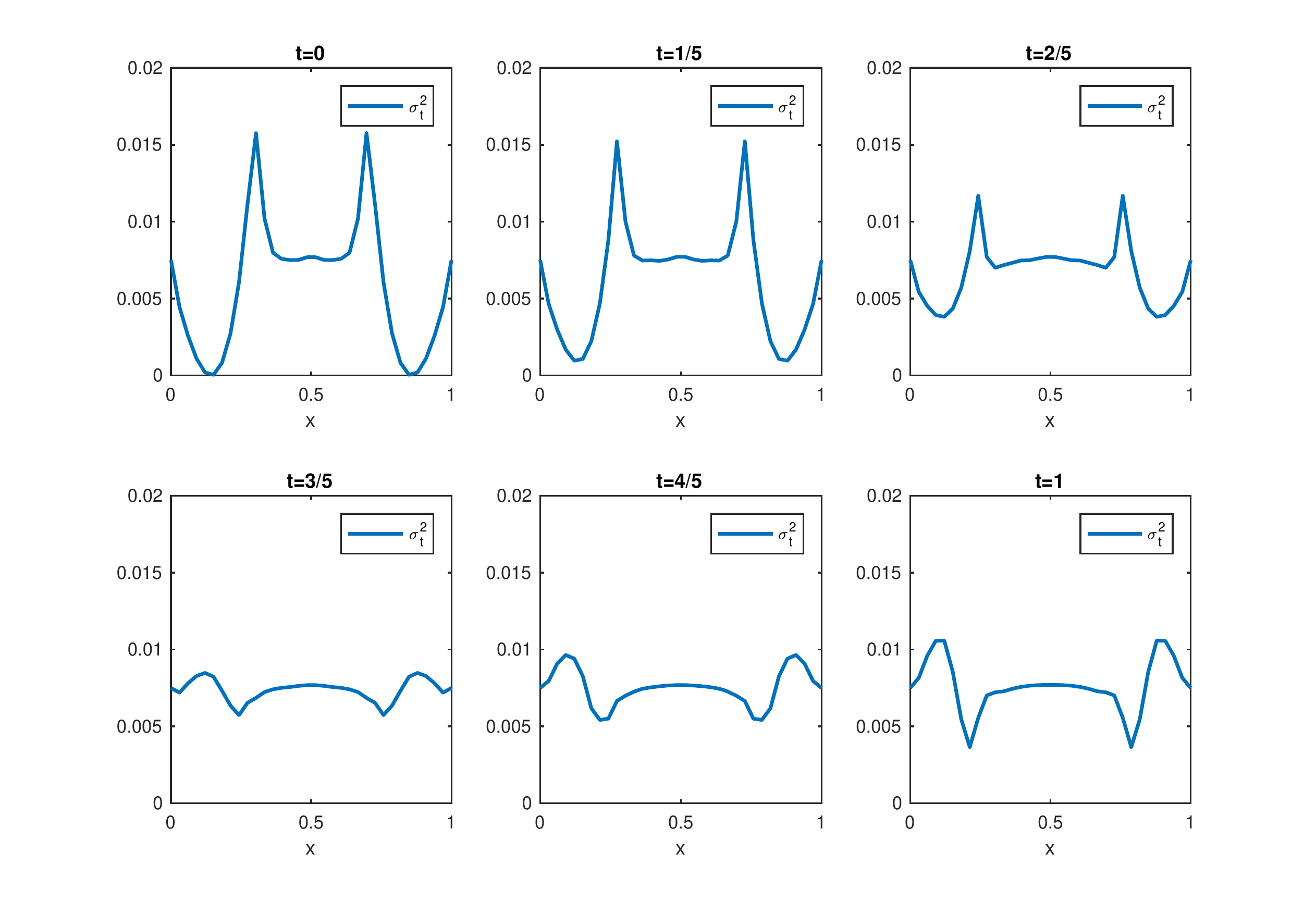}
\caption{The variance $\sigma^2(t,x)$.}
\label{fig2}
\end{figure}
\begin{figure}[t]
\includegraphics[width=\textwidth]{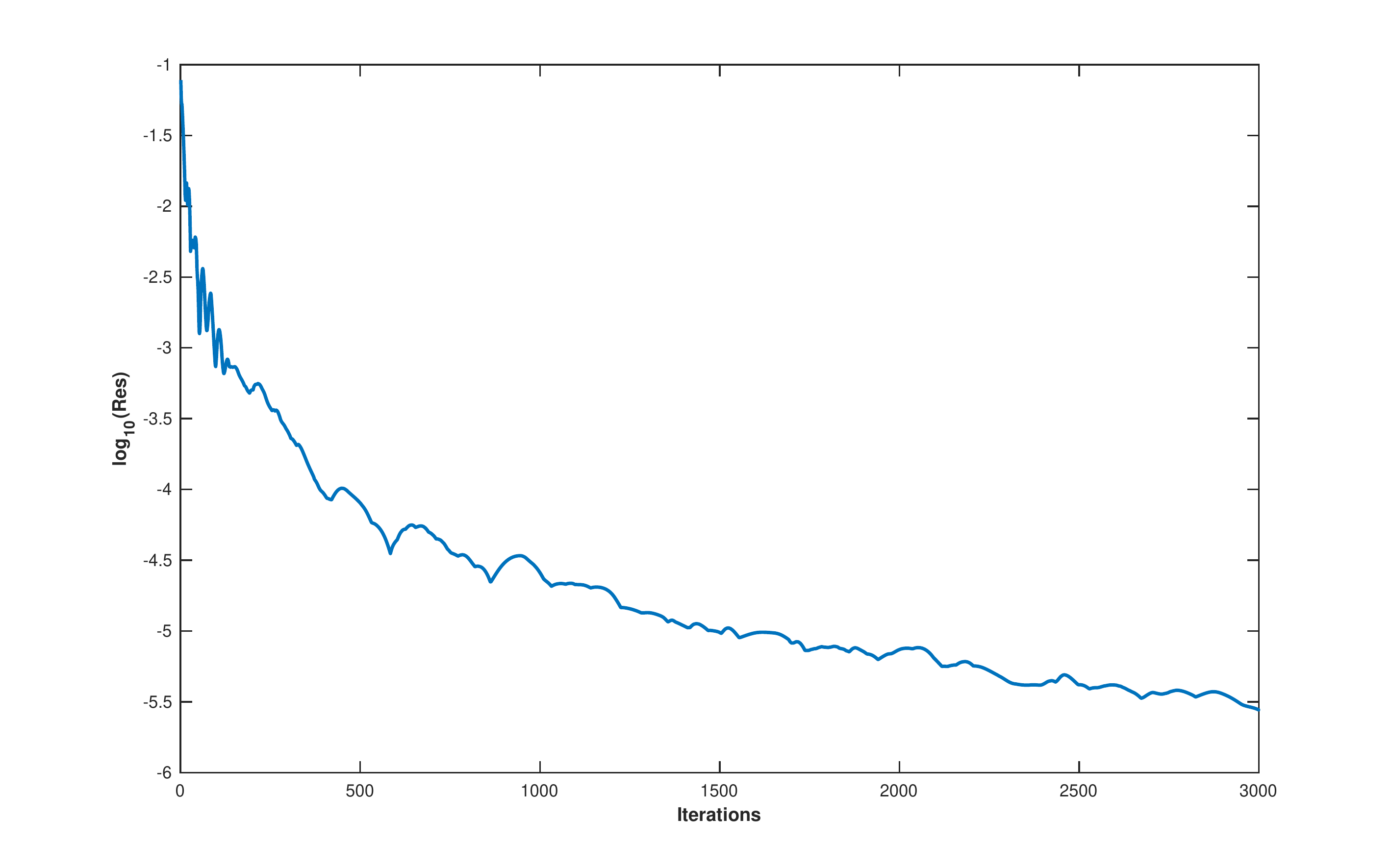}
\caption{The residual $res^n$.}
\label{fig3}
\end{figure}

The algorithm was implemented and tested on the following simple example. Consider the computational domain $x\in[0,1]$ and the time interval $t\in[0,1]$. We set the initial and final distributions to be $X_0\sim N(0.5,0.05^2)$ and $X_1\sim N(0.5,0.1^2)$ respectively, where $N(\mu,\sigma^2)$ denotes the normal distribution. The following cost function was chosen:
\begin{equation}
  F\left(\gamma\right) = \begin{cases}
  (\gamma - \overline{\gamma})^2, &\gamma\geq 0, \\
  +\infty, &\text{otherwise,}
  \end{cases}
\end{equation}
where $\overline{\gamma}$ was set to 0.00375 so that the optimal value of variance is constant $\sigma^2 = 0.1^2-0.05^2=0.0075$. Then we discretized the space-time domain as a $128\times128$ lattice. The penalization parameter is set to $r=64$. The results after 3000 iterations are shown in Figures \ref{fig1} and \ref{fig2}, and the convergence of the residuals is shown in figure \ref{fig3}. The convergence speed decays quickly, but we reach a good approximation after about 500 iterations. The noisy tails in Figure \ref{fig2} correspond to regions where the density $\rho$ is close to zero. The diffusion process has a very low probability of reaching these regions, so the value of $\sigma^2$ has little impact. 
%The unstable tails in Figure \ref{fig2} are numerical errors caused by small values of $\rho$ when computing the variance $\sigma^2=2m/\rho$. 
In areas where $\rho$ is not close to zero, $\sigma^2$ remains constant which matches the analytical solution. 
%Alternatively, $\sigma^2$ can be calculated by $\partial_{xx}\phi$, which gives a smooth surface.

%%%%%%%%%%%%%%%%%%%%%%%%%%%%%%%%%%%%%%%%%%%%%%%%%%%%%%%%%%%%%%%%

\section{Summary}

This paper focuses on a new approach for the calibration of local volatility models. Given the distributions of the asset price at two fixed dates, the technique of optimal transport is applied to interpolate the distributions and recover the local volatility function, while maintaining the martingale property of the underlying process. Inspired by \cite{MR1738163}, the problem is first converted into a saddle point problem, and then solved numerically by an augmented Lagrangian approach and the alternative direction method of multipliers (ADMM) algorithm. The algorithm performs well on a simple case in which the numerical solution matches its analytical counterpart.
The main drawback of this method is due to the slow convergence rate of the ADMM algorithm. We observed that a higher penalization parameter may lead to faster convergence. Further research is required to conduct more numerical experiment, improve the efficiency of the algorithm and apply it to more complex cases.

\bigskip
% BibTeX users please use
\bibliographystyle{spmpsci}
\bibliography{ot}

\end{document}